\documentclass[11pt]{article}
\usepackage[utf8]{inputenc}
\usepackage[T1]{fontenc}
\usepackage{lmodern}
\usepackage{microtype}
\usepackage{amsmath,xcolor}
\usepackage{amssymb}
\usepackage{amsthm}
\usepackage{bm,float}
\usepackage{dsfont}
\usepackage{fullpage,caption,wrapfig}
\usepackage{comment}
\usepackage{mathtools}
\usepackage[shortlabels]{enumitem}
\usepackage{complexity}
\usepackage[backend=bibtex, style=alphabetic, backref=true,url=false, giveninits=true,maxcitenames=5, maxbibnames=99]{biblatex} 
\renewbibmacro{in:}{%
  \ifentrytype{article}{}{\printtext{\bibstring{in}\addcolon\space}}}

\renewbibmacro*{volume+number+eid}{%
  \printfield{volume}%
  \printfield{number}%
  \setunit{\addcomma\space}%
  \printfield{eid}}
\DeclareFieldFormat[article]{number}{\mkbibparens{#1}}

\renewbibmacro*{issue+date}{%
  \setunit{\addperiod\space}
    \iffieldundef{issue}
      {\usebibmacro{date}}
      {\printfield{issue}%
       \setunit*{\addspace}%
       \usebibmacro{date}}
  \newunit}

\AtEveryBibitem{
 \clearlist{address}
 \clearlist{location}
 \clearfield{month}
 \ifentrytype{inproceedings}{\clearname{editor}}{}
}

\usepackage{bbm}

\usepackage{nag,tikz}
\usetikzlibrary{calc}
\usetikzlibrary{decorations.pathreplacing}
\usetikzlibrary{shapes, patterns, decorations, fit, intersections, arrows, automata, positioning}

\usepackage{algorithm,algpseudocode}
\usepackage{multicol}
\usepackage{thmtools} 
\usepackage{hyperref}
\hypersetup{
    colorlinks=true,
    linkcolor=violet,
    filecolor=magenta,      
    urlcolor=cyan,
    citecolor=blue,
    pdffitwindow=true,
}\usepackage[capitalize, nameinlink]{cleveref}

\algrenewcommand\algorithmicrequire{\textbf{Input:}}
\algrenewcommand\algorithmicensure{\textbf{Output:}}

\theoremstyle{plain}
\newtheorem{theorem}{Theorem}[section]
\newtheorem{lemma}[theorem]{Lemma}

\theoremstyle{definition}
\newtheorem{definition}[theorem]{Definition}

\theoremstyle{remark}

\newcommand{\emphdef}[1]{\emph{#1}}

\newcommand{\old}[1]{}

\newcommand{\rank}{\text{rank}}

\DeclareMathOperator{\polyloglog}{polyloglog}

\newcommand{\Rplus}{\ensuremath{\mathbb R}_{\geq 0}}
\renewcommand{\R}{\ensuremath{\mathbb R}}
\newcommand{\Z}{\ensuremath{\mathbb Z}}

\newcommand{\lamconstrainedbasis}{\textsc{LamConstrainedBasis}}
\DeclareMathOperator{\Span}{span}
\newcommand{\indset}[1]{\chi(#1)}

\def\b1{{\bf 1}}
\def\1{{\bf 1}}

\def\cI{{\cal I}}

\def\cL{{\cal L}}

\def\R{\mathbb{R}}

\def\cP{{\cal P}}
\def\cQ{{\cal Q}}
\def\cI{{\cal I}}

\newcommand{\matroid}{\mathcal{M}}
\newcommand{\baseP}[1]{P_{#1}}

\newcommand{\Ldeg}{\mathcal{L}}

\newcommand{\Ltight}{\mathcal{L}_{\mathrm{tight}}}
\newcommand{\LdegB}{\mathcal{L}_{\mathrm{basis}}}

\newcommand{\Cbasis}{\mathcal{C}_{\mathrm{basis}}}
\newcommand{\Pst}{P_{\mathrm{st}}}
\newcommand{\Pmat}{P_{\Matroid}}

\newcommand{\Pstdom}{\Pst^{\uparrow}}
\newcommand{\Matroid}{\mathcal{M}}
\DeclareMathOperator{\rnk}{rank}

\newcommand{\symdiff}{\mathbin{\triangle}}
\newcommand{\delt}[2]{|#2 \cap \delta(#1)|}

\renewcommand{\smallsetminus}{\setminus}

\newcommand{\deltaP}{\delta}

\newcommand{\declareperson}[1]{\expandafter\newcommand\csname#1\endcsname[1]{\begingroup\em\textcolor{orange}{#1: ##1}\endgroup}}

\declareperson{Nathan}
\hypersetup{ colorlinks=true, linkcolor=black!40!blue, filecolor=magenta, urlcolor=blue, }

\definecolor{arylideyellow}{rgb}{0.91, 0.84, 0.42}
\definecolor{darkgreen}{rgb}{0.01, 0.75, 0.24}
\definecolor{darkred}{rgb}{0.76, 0.13, 0.28}

\begin{document}

\title{Thin trees for laminar families}
\author{Nathan Klein\thanks{\href{mailto:nwklein@cs.washington.edu}{nwklein@cs.washington.edu}. Research supported in part by Air Force Office of Scientific Research grant FA9550-20-1-0212 and NSF grants DGE-1762114 and  CCF-1813135.}\\ University of Washington \and  Neil Olver\thanks{\href{mailto:n.olver@lse.ac.uk }{n.olver@lse.ac.uk}. Research supported by NWO Vidi grant 016.Vidi.189.087.}\\ London School of Economics\\ and Political Science}


\maketitle

\begin{abstract}
In the laminar-constrained spanning tree problem, the goal is to find a minimum-cost spanning tree which respects upper bounds on the number of times each cut in a given laminar family is crossed. 
This generalizes the well-studied degree-bounded spanning tree problem, as well as a previously studied setting where a chain of cuts is given. 
We give the first constant-factor approximation algorithm;  
        in particular we show how to obtain a multiplicative violation of the crossing bounds of less than 22 while losing less than a factor of 5 in terms of cost.   

    Our result compares to the natural LP relaxation.
    As a consequence, our results show that 
    given a $k$-edge-connected graph and a laminar family $\cL \subseteq 2^V$ of cuts, there exists a spanning tree which contains only an $O(1/k)$ fraction of the edges across every cut in $\cL$. 
    This can be viewed as progress towards the \emph{Thin Tree Conjecture}, which (in a strong form) states that this guarantee can be obtained for all cuts simultaneously.
\end{abstract}

\section{Introduction}\label{sec:introduction}

Let $G=(V,E)$ be a connected undirected graph. 
Given any proper $S \subset V$, we use $\delta(S)$ to denote the cut with shores $S$ and $V \setminus S$.
A spanning tree $T$ of $G$ is called \emph{$\alpha$-thin} if the number of edges of $T$ crossing any given cut of $G$ is at most an $\alpha$ fraction of the total number of edges: $|T \cap \delta(S)| \leq \alpha |\delta(S)|$ for each $S \subseteq V$.

In 2004, Goddyn~\cite{God04} made the following conjecture: there exists a function $f: \Z_+ \to [0,1]$ with $\lim_{k \to \infty} f(k)/k = 0$ such that every $k$-edge-connected graph $G$ has an $f(k)$-thin spanning tree.
This has become known as the \emph{thin tree conjecture}, and it remains open despite substantial efforts.

A natural strengthening of the conjecture, which we will refer to as the \emph{strong thin tree conjecture} makes the same claim, but for $f(k) = C/k$ for some constant $C$. This conjecture is found explicitly in \cite{AGMOS17}
and is the best that one could hope for up to constant factors; clearly no $k$-edge-connected graph has an $\alpha$-thin tree for any $\alpha < 1/k$.
In a different direction, there is also an algorithmic question one can ask: if a thin tree always exists, can we \emph{find} one in polynomial time? 

The thin tree conjecture has some nice implications. It implies the \emph{weak 3-flow conjecture} of Jaeger~\cite{Jae84}.
This has since been resolved, by Thomassen~\cite{Thom12}, however this would provide an alternate proof. 
Another application lies in the asymmetric traveling salesman problem (ATSP). As shown by Asadpour, Goemans, Madry, Oveis Gharan and Saberi~\cite{AGMOS17,OS11}, if the constructive form of the strong thin tree conjecture is true, 
it would yield an $O(1)$-approximation algorithm to ATSP. 
This has since been resolved by Svensson, Tarnawski and V\'egh~\cite{STV18} using completely different methods.
Nonetheless, a new algorithm stemming from thin trees would be of significant interest. 
Furthermore, a constant factor approximation algorithm to the \emph{bottleneck} version of the asymmetric traveling salesman problem, where the goal is to minimize the longest edge in the tour rather than the sum, is not known. 
This would follow from the constructive form of the thin tree conjecture~\cite{AKS21}. 
 

Although the (strong) thin tree conjecture would no longer imply breakthroughs to these other problems, it remains a natural question in its own right.
Turning things around, the positive resolution of these implications can perhaps be viewed as some weak evidence for the conjecture.

For the following discussion, it is useful to observe that the strong thin tree conjecture has the following equivalent formulation.
Suppose we are given a graph $G$ as well as a point $x$ in the spanning tree polytope (that is, a convex combination of characteristic vectors of edge sets of spanning trees of $G$).
We say that a spanning tree $T$ is \emph{$\alpha$-thin with respect to $x$} if $|T \cap \delta(S)| \leq \alpha x(\delta(S))$ for every $S \subseteq V$.
The conjecture is that there is a universal constant $\alpha$ such that an $\alpha$-thin tree with respect to $x$ always exists, for any instance and point in the spanning tree polytope.
The equivalence follows from the observation that the point $x'$ defined by $x'_e = 2/k$ for all $e \in E$ is in the dominant of the spanning tree polytope for every $k$-edge-connected graph $G$, and so there is a point $x$ in the spanning tree polytope with $x_e \leq 2/k$ for all $e$.
An $\alpha$-thin tree with respect to $x$ is then a $(2\alpha/k)$-thin tree for the graph.

\paragraph{Progress on the thin tree conjecture.}
The conjecture is known to hold for some graph classes, most notably planar and bounded genus graphs~\cite{OS11}.
For general graphs, the best known result is that there always exists an $O(\frac{\polyloglog n}{k})$-thin tree in any $k$-connected graph \cite{AO15}.
This is non-constructive; constructively, the best known is only $O(\frac{\log n}{\log \log n\cdot k})$-thinness.

One difficulty with the constructive form of the conjecture is that it's not even clear how to \emph{check} if a given tree $T$ is $\alpha$-thin, or even $O(\alpha)$-thin. 
Nor do we know of a polynomially checkable certificate that can certify thinness.
The problem, of course, is that there are an exponential number of cuts to be concerned with. 
An easier question presents itself: what if we consider an explicitly given family of cuts, and require the thinness condition $|T \cap \delta(S)| \leq \alpha |\delta(S)|$ only for these specific cuts?
And one step further: what if we consider a family of cuts with some specific structure? 

\paragraph{Explicitly given cut collections.}
Related questions have been considered from an algorithmic perspective already, independently from the thin tree conjecture.
The first class considered was that of \emph{singleton cuts}.
Suppose we are given an integer-valued degree bound $b_v$ for each node $v$ of the graph $G$. 
The degree bounded spanning tree problem asks for a spanning tree satisfying these bounds, if such a spanning tree exists.
This problem is easily seen to be NP-hard, since it captures the question of finding a Hamiltonian path with a specified start and end node.
So it is necessary to allow for some relaxation of the degree bounds.
F\"urer and Raghavachari~\cite{FR92} showed that relaxing the degree bounds by $1$ additively suffices.
That is, they showed how to efficiently find a spanning tree $T$ satisfying $|T \cap \delta(v)| \leq b_v + 1$ for all $v \in V$, if there exists a spanning tree $T^*$ that satisfies the degree bounds exactly. 

One can also consider a minimum cost version of the question.
Now each edge $e \in E$ has a nonnegative cost $c(e)$, and the goal is to find a \emph{cheapest} spanning tree satisfying the degree bounds (again, assuming one exists).
Goemans~\cite{Goe06} showed how to efficiently find a spanning tree $T$ which violates the degree bounds by at most an additive 2, and satisfies $c(T) \leq c(T^*)$, where $T^*$ is a minimum cost spanning tree that satisfies all the degree bounds exactly.
Singh and Lau~\cite{SL15} then showed how to improve the degree violation to just $1$, while maintaining the same bound on the cost.
They use the method of \emph{iterative relaxation}; we use iterative relaxation as well, so we will discuss this further in the sequel.

That ends the story for degree bounds; what about other families of constraints?
So we have a given family $\mathcal{F}$ of subsets, and a ``degree bound'' $b_S$ for each $S \in \mathcal{F}$.
Olver and Zenklusen~\cite{OZ13} showed how to obtain, constructively, a constant multiplicative violation of all cut constraints if $\mathcal{F}$ is a \emph{chain}; that is, $\mathcal{F} = \{S_1, S_2, \ldots, S_t\}$ with $S_1 \subsetneq S_2 \subsetneq \cdots \subsetneq S_t$.
Linhares and Swamy~\cite{LS16} showed that a minimum cost version of this result also holds, if one allows a constant factor approximation in the cost as well as in the cut constraints.

All of these results compare to the natural fractional relaxation.
That is, they do not require that there is an actual tree satisfying the degree bounds, merely that there is a point in the spanning tree polytope which does.
As such, we can view them in the context of thin trees. 
They show that weaker versions of the strong thin tree conjecture hold, where the cut bounds are enforced only on singleton cuts, or only on a chain of cuts.
We will say that the strong thin tree conjecture holds for a given family $\mathcal{F}$ if given any $x$ in the spanning tree polytope, there is a spanning tree $T$ satisfying $|T \cap \delta(S)| \leq O(1) x(\delta(S))$ for all $S \in \mathcal{F}$.

%


\paragraph{Our results.}
Given that the strong thin tree conjecture holds for the family of singletons, and for a chain family, a very natural question presents itself. 
Suppose that $\Ldeg$ is an arbitrary \emph{laminar family} of subsets of $V$; that is, for every $S,T \in \Ldeg$, $S \cap T$
 is either equal to $\emptyset,S$, or $T$.
Does the strong thin tree conjecture hold for $\Ldeg$?

We show that this is indeed true.
Further, our proof is constructive, and allows for costs.
More precisely, given arbitrary nonnegative edge costs, our returned tree has cost within a constant factor of the cost of the starting fractional solution $x$.


We briefly sketch our high-level approach, leaving a full overview until \Cref{sec:overview}.
As already mentioned, iterative relaxation has been applied very successfully to the degree-bounded spanning tree problem, so it is a natural candidate approach.
However, there is an immediate obstruction.
Iterative relaxation for degree-bounded spanning tree is fairly insensitive to the use of the graphic matroid; it works just as well (essentially without changes) if the graphic matroid is replaced by any other matroid.\footnote{With the notable exception of \cite{LS16} which solves the bounded degree spanning problem with an additive error of 1 on both lower \textit{and} upper bounds. When translated to the general matroid setting, the additive error is only known to be 2 \cite{KLS08}.}
However, the matroid generalization of the laminar-constrained spanning tree problem does \emph{not} have a constant integrality gap, and even obtaining a constant factor multiplicative violation is hard. This was shown by Olver and Zenklusen~\cite{OZ13} already for the chain case.
So any successful approach will need to exploit the graphic matroid specifically; it is not clear how to do this directly with iterative relaxation.

We manage to bypass this obstruction and make use of iterative relaxation. 
We do this by first reducing to a special class of instances that we call $\Ldeg$-\emph{aligned}, where the fractional solution $x$ has the property that for every set $S$ in the laminar family of constraints $\Ldeg$, the restriction of $x$ to $S$ is a point in the base polytope of the graphic matroid for the graph restricted to $S$. 
Our reduction crucially exploits properties of spanning trees, and does not apply to general matroids.
We then give an iterative relaxation proof of this $\Ldeg$-aligned case.
This part \emph{does} generalize to arbitrary matroids.

\paragraph{Other related work.} For laminar families, the most directly comparable work is from 2013 by Bansal, Khandekar, Könemann, Nagarajan, and Peis. They give an \textit{additive} $O(\log n)$ approximation for the laminar constrained spanning tree problem \cite{BKKNP13}, improving upon an earlier more general result which given a family of $m$ constraints obtains a violation of $(1+\epsilon)b+O(\frac{1}{\epsilon}\log m)$ for each bound $b$ \cite{CVZ10}.
As previously mentioned, Olver and Zenklusen~\cite{OZ13} demonstrated a constant factor multiplicative violation for a family of cuts given by a chain. 
These three results are with respect to the fractional relaxation, and thus also solve the related thin tree problems. 
N\"agele and Zenklusen~\cite{NZ19} demonstrated that in quasi-polynomial time the violation for the chain-constrained spanning tree problem can be improved to a $(1+\epsilon)$ multiplicative factor, for any $\epsilon > 0$.
They further generalize this slightly towards laminar families, by allowing for a family of cuts that form a laminar family of constant \emph{width}, meaning that the maximum number of disjoint sets in the laminar family is bounded by a constant. 
(Put differently, the number of leaves in the tree representing the laminar family is constant).
However, this result is not based on the LP relaxation, and so does not imply anything for the strong thin tree conjecture for chains or constant-width laminar families.

This problem has also been studied for general matroids. 
Kir\'aly, Lau and Singh~\cite{KLS08} showed that given a matroid $\Matroid$ and a collection of upper bound constraints, one can achieve an additive violation of $\Delta-1$ for all constraints, so long as every element of the matroid is in at most $\Delta$ constraints. They achieve a similar guarantee if lower bounds (or both lower and upper bounds) are present. Similar results and further  generalizations can be found in \cite{CVZ10, BKKNP13}.

Pritchard~\cite{Pri11} conjectured that every $k$-edge-connected graph contains a spanning tree after whose deletion the graph remains $k-f(k)$ connected, where $f(k)$ is any function for which $\lim_{k \to \infty} f(k)/k=0$. This can easily be seen as a weakening of the thin tree conjecture. The strong version of this conjecture (which is a consequence of the strong thin tree conjecture) is that $f(k)$ is an absolute constant. Currently the best known bound for this problem (to the best of our knowledge) is $f(k) = \lfloor \frac{k}{2} \rfloor -1$ by the Nash-Williams theorem \cite{NW61}.


There is a natural \emph{spectral} strengthening of thin trees.
Let $L_H$ denote the Laplacian of a graph $H$, and let $\preceq$ denote the L\"owner ordering on symmetric matrices\footnote{That is, $A \preceq B$ if $B-A$ is positive semidefinite.}.
We say $T$ is \emph{$\alpha$-spectrally-thin} if $L_T \preceq \alpha L_G$; that is, if $z^T L_Tz \leq \alpha z^T L_G z$ for any vector $z \in \mathbb{R}^V$.
This is a stronger condition than $\alpha$-thinness, as can be seen by choosing $z$ to be the characteristic vector of a set $S \subseteq V$.
A big advantage of spectral thinness is that it can be efficiently checked.
A natural analogue of the strong thin tree conjecture, where connectivity is replaced by the minimum effective conductance, can be derived~\cite{HO14} as a consequence of results on the the Kadison-Singer problem~\cite{MSS13}.
This demonstrates that the strong thin tree conjecture holds for edge transitive graphs (or any graph where the minimum edge conductance is within a constant factor of the connectivity).
Unfortunately, spectral thinness is too strong a property to directly aid in proving the (strong or weak) thin tree conjecture in general; there are instances where no $o(\sqrt{n}/k)$-spectrally thin tree exists~\cite{HO14,Goe12}.
Nonetheless, spectral approaches have been fruitful.
The current best result by Anari and Oveis Gharan~\cite{AO15} mentioned previously, that $O(\log \log n/k)$-thin trees exist, makes use of spectral methods in a sophisticated way.
Our approach on the other hand is completely combinatorial; we will not make use of any spectral techniques.

\section{Preliminaries and Results}\label{sec:preliminaries}

\subsection{Notation}
Given a graph $G=(V,E)$ and a subset $S\subseteq V$, let $\delta(S)= \{\{u,v\}: |\{u,v\}\cap S|=1\}$ denote the set of edges with exactly one endpoint in $S$. Let $G[S]$ denote the induced graph of $G$ whose vertex set is $S$, and let $E(S) \subseteq E$ denote the set of edges in $G[S]$. 
For $\cP=\{P_1,\dots,P_k\}$ a partition of a subset of the vertices of $G$, we let $\deltaP(\cP)$ denote the set of edges with endpoints in two different sets $P_i$. 
If the choice of $G$ is not clear, we may write, e.g., $\delta_G(S)$ or $\deltaP_G(\cP)$.


For any edge weight function $x: E \to \R$, we write $x(F) := \sum_{e\in F} x(e)$. For $F \subseteq E$, we write $x_{\mid F}$ to denote $x$ restricted to $F$. 

\subsection{Polyhedral Background} 

Edmonds \cite{Edm70} gave the following description for the convex hull of the spanning trees of any graph $G=(V,E)$, known as the {\em spanning tree polytope}.
\hypertarget{stpolytope}{\begin{equation}
        \Pst(G) = \bigl\{ x \in \Rplus^E : x(E) = |V|-1, \; x(E(S)) \leq |S|-1 \; \forall S \subseteq V\bigr\}.
\label{eq:spanningtreelp}
\end{equation}}
The following is the natural LP relaxation for the problem given in \cref{def:problem}. 
\begin{equation}\label{eq:lp}
\begin{aligned}
	\min \quad& \sum_{e \in E} x_e c_e& \\
	\text{s.t.,} \quad & x(\delta(S)) \leq b_S &\forall S\in \Ldeg,\\
	& x \in \Pst(G) \\
\end{aligned}	
\end{equation} 


For two  $x,x' \in \R^{E}$, we say $x$ \textbf{dominates} $x'$ if $x-x' \ge 0$. 
\hypertarget{pstdom}{Let $\Pstdom(G)$ denote the dominant of the spanning tree polytope of $G$, that is, the set of points in $x \in \R^{E}$ which dominate some point in $\Pst(G)$.} $\Pstdom(G)$ has the following characterization:
\hypertarget{stpolytope}{\begin{equation}
        \Pstdom(G) = \bigl\{ x \in \Rplus^E : x(\deltaP(\cP)) \geq |\cP|-1 \quad \forall \text{ partitions } \cP \text{ of } V \bigr\}.
\label{eq:spanningtreelp}
\end{equation}}
%
It is well-known that $\Pstdom(G)$ can be separated efficiently.
\begin{theorem}[\cite{Bar92}]\label{thm:domsep}
    Given a graph $G=(V,E)$ and a point $x \in \Rplus^{E}$, A partition $\cP$ of $G$ minimizing $x(\deltaP(\cP)) - (|\cP|-1)$ can be found in polynomial time.
\end{theorem}

\medskip

Suppose $\matroid = (E, \cI)$ is a matroid with groundset $E$ and independent sets $\cI$. 
The \emph{matroid base polytope} of $\matroid$, which we will denote $\baseP{\matroid}$, is the convex hull of the incidence vectors of all bases of $\matroid$. The \emph{rank} of $\matroid$, denoted $\rnk_{\matroid}$, is the cardinality of the largest independent set of $\matroid$. 
Given $F \subseteq E$:
\begin{enumerate}
\item The \emph{deletion} of $F$ from $\matroid$ is the matroid on the groundset $\matroid \smallsetminus F$ with independent sets $\{I \smallsetminus F : I \in \cI\}$. If $F=\{e\}$, i.e. it is a singleton, we will use the shorthand $M - e$. 
\item The \emph{restriction} of $\matroid$ to $F$,  denoted $\matroid_{\mid F}$, is the matroid on groundset $F$ with independent sets $\{ I \cap F : I \in \cI\}$. This is equivalent to the deletion of $E \setminus F$.
\item The \emph{contraction} of $\matroid$ by $F$, denoted $\matroid / F$, is the matroid on the groundset $E \setminus F$ with independent sets $\{ I \subseteq E \setminus F : I \cup B \in \cI\}$, where $B$ is an arbitrary basis of $\matroid_{\mid F}$ (equivalently, an independent set of $\matroid$ contained in $F$ of largest cardinality). If $F=\{e\}$, i.e. it is a singleton, we will use the shorthand $\matroid/e$. 
\end{enumerate}





\subsection{Our Results}\label{sec:results}
We recall that a family of sets $\cL \subseteq 2^V$ is \textbf{laminar} if for all $S,T \in \cL$, $S \cap T$ is either equal to $\emptyset,S$, or $T$. 



\begin{definition}[Laminar constrained spanning tree problem]\label{def:problem}
	Let $G=(V,E)$ be a connected graph, and $\Ldeg$ a laminar family on $V$, with an associated degree bound $b_S \in \mathbb{Z}_{\geq 0}$ for each $S \in \Ldeg$. The goal is to find efficiently a spanning tree $T$ for which 
        $\delt{S}{T} \leq \alpha b_S$
        for each $S \in \Ldeg$, assuming that there does exist a spanning tree $T^*$ satisfying $\delt{S}{T^*} \leq b_S$. In such a case, we say $T$ is an $\alpha$-approximate solution to the laminar constraints.
We assume for convenience that $V \in \Ldeg$, though the associated constraint is of course vacuous.
\end{definition}
To solve the above problem, we first determine if LP \eqref{eq:lp} is feasible, which can be done in polynomial time. If it is not, we may return ``no'' to the above problem since this would certify that such a tree does not exist. Thus to obtain an $\alpha$ approximation for the problem above, it is enough to obtain an $\alpha$-thin tree with respect to a solution $x$ of \eqref{eq:lp}.
\hypertarget{problemstatement}{\begin{definition}[Laminar $\alpha$-thin tree for $(G,\Ldeg,x)$] 
As input we get a graph $G=(V,E)$, a laminar family $\Ldeg$ over $V$, and a feasible LP solution $x$ to LP \eqref{eq:lp}. 
Our goal is to find a spanning tree $T$ such that $\delt{S}{T} \le \alpha x(\delta(S))$ for all $S \in \Ldeg$, i.e. a tree that is $\alpha$ thin with respect to $x$.
\end{definition}}

The main result of this paper is the following. We remark it also gives an $O(1)$ approximation in terms of the cost of the tree.
\begin{theorem}\label{thm:main}
Given an instance \hyperlink{problemstatement}{$(G, \Ldeg, x)$}, we can in polynomial time find a spanning tree $T$ such that:
\begin{enumerate}[i)]
\item $c(T) \le (2+\sqrt{7})c(x) < 5c(x)$, and
\item $\delt{S}{T} \le (2+\sqrt{7})^2x(\delta(S)) < 22x(\delta(S))$, i.e., it is a 22-thin tree for $(G,\Ldeg,x)$. 
\end{enumerate}
\end{theorem}

Our theorem can be generalized as follows, which can be used to reduce the cost of the tree arbitrarily close to 2 (at the expense of incurring a larger multiplicative loss).
\begin{theorem}[Main]\label{thm:main2}
Given an instance \hyperlink{problemstatement}{$(G, \Ldeg, x)$} and any $\eta > 2$, we can in polynomial time find a spanning tree $T$ such that:
\begin{enumerate}[i)]
\item $c(T) \le \eta c(x)$, and
\item $\delt{S}{T} \le \frac{1}{1-\frac{2}{\eta}}(2\eta + 3)x(\delta(S))$.
\end{enumerate}
\end{theorem}
We will prove the latter theorem, since the previous theorem follows by setting $\eta = 2+\sqrt{7}$. 

\subsection{Proof Overview}\label{sec:overview}

A key observation of this paper is the usefulness of the following definition.
\hypertarget{tight}{\begin{definition}[$\Ldeg$-aligned]\label{def:Laligned}
Given a graph $G=(V,E)$ and a laminar family $\Ldeg \subseteq 2^V$ we say a point $x \in \Pst(G)$ is \emphdef{$\Ldeg$-aligned} if $x_{\mid E(S)} \in \Pst(G[S])$ for all $S \in \Ldeg$.
\end{definition}}
Note that $G[S]$ should be connected for each set $S \in \Ldeg$, otherwise no point can be $\Ldeg$-aligned.

In \cref{sec:reduction} and \cref{sec:rounding} we show the following two theorems which when combined immediately give \cref{thm:main2}. 
\begin{restatable}[Laminar thin trees for $\Ldeg$-aligned points]{theorem}{rounding}\label{thm:rounding}
Given an instance \hyperlink{problemstatement}{$(G, \Ldeg, x)$} for which $x$ is \hyperlink{tight}{$\Ldeg$-aligned}, we can find a tree $T$ of cost at most $c(x)$ in polynomial time for which $$\delt{S}{T} \le 2\lceil x(\delta(S))\rceil + 1 \le 2x(\delta(S))+3$$ for all $S \in \Ldeg$. 
\end{restatable} 

\begin{restatable}[Reduction to $\Ldeg$-aligned points]{theorem}{reduction}\label{thm:reduction}
    For any instance \hyperlink{problemstatement}{$(G, \Ldeg, x)$} and any $\eta > 2$, we can find an instance $(G, \Ldeg', x')$ in polynomial time such that:
    \begin{enumerate}[i)]
    \item $x'$ is \hyperlink{tight}{$\Ldeg'$-aligned},
    \item $x'$ is dominated by $\eta x$,
    \item If for a spanning tree $T$ there are $\alpha,\beta  \ge 0$ such that $\delt{S}{T} \le \alpha x'(\delta(S)) + \beta$ for all $S \in \Ldeg'$, then we have $\delt{S}{T} \le \frac{1}{1-\frac{2}{\eta}}(\eta \alpha + \beta)x(\delta(S))$ for all $S \in \Ldeg$. 
    \end{enumerate}
\end{restatable} 

To obtain our main theorem, given an instance $(G,\Ldeg)$ we solve LP \eqref{eq:lp} to obtain an instance $(G,\Ldeg,x)$. We then apply \cref{thm:reduction} to obtain a $\Ldeg$-aligned instance $(G,\Ldeg',x')$. Finally, we apply \cref{thm:rounding} to obtain our tree with the desired properties.

We remark that while \cref{thm:rounding} can be generalized to hold for any matroid over the edges of a graph and any laminar family (see \cref{sec:rounding}), \cref{thm:reduction} cannot be.
Olver and Zenklusen \cite{OZ13} showed that there is a matroid and a laminar family of constraints (in fact, their family is a chain, and their matroid simply a partition matroid) with no constant-thin basis, in particular giving a lower bound of $O(\frac{\log n}{\log \log n})$ on the multiplicative violation. Thus it is necessary that one of these two pieces cannot be generalized to all matroids.

\cref{thm:reduction} is proved via a natural combinatorial procedure which iteratively replaces sets in $\Ldeg$ that are far from meeting the criteria $x_{\mid E(S)} \in \Pst(G[S])$ with some partition of them. We first consider the scaling $\eta x$, and show that if $\eta \cdot x_{\mid E(S)} \in \Pstdom(G[S])$ for all $S \in \Ldeg$, then there is a point $x'$ dominated by $\eta x$ which is $\Ldeg$-aligned. If not, we iteratively find a minimal cut $S$ for which $\eta \cdot x_{\mid E(S)} \notin \Pstdom(G[S])$, and then find the partition $\cP=\{P_1,\dots,P_k\}$ of $S$ which maximally violates an inequality in $\Pstdom(G[S])$. We then delete $S$ from the laminar family and add $P_1,\dots,P_k$. We show that $\eta \cdot x_{\mid E(P_i)} \in \Pstdom(G[P_i])$ for all $i$. Therefore, by applying this procedure we get closer to obtaining an $\Ldeg$-aligned point. To finish the proof, we show that this process allows us to still effectively maintain (iii) of \cref{thm:reduction}.


\cref{thm:rounding} (and its generalization to arbitrary matroids) is proved via an iterative relaxation procedure. The criteria that $x$ is $\Ldeg$-aligned is, in some sense, exactly what is needed to make the iterative relaxation procedure work.

\section{Reduction to $\Ldeg$-aligned Points}\label{sec:reduction}

\begin{figure}
\centering
\begin{tikzpicture}
\draw[color=darkred, thick] (0,0) node[above left=1.65cm and 1.65cm] {$S$} circle (2.5cm);
\draw (-1.2,0) node[ellipse, draw=darkgreen, text=darkgreen, thick, minimum width = 1.9cm, 
    minimum height = 3.5cm] (C) {$P_1$};
\draw (1.2,0) node[ellipse, draw=darkgreen, text=darkgreen, thick, minimum width = 1.9cm, 
    minimum height = 3.5cm] (C) {$P_2$};
\draw[-, thick] (-0.5,0.3) -- (0.5,0.3) node[above left=0.2cm and 0.05cm] {$< \frac{1}{\eta}$};
\draw[-, thick] (-0.5,0) -- (0.5,0) node {};
\draw[-, thick] (-0.5,-0.3) -- (0.5,-0.3) node {};
\end{tikzpicture}
\caption{An example of a set which is not $\eta$-well-connected (see \cref{def:wellconnected}). In this case, \cref{alg:reduction} may replace $S$ by $P_1$ and $P_2$ in $\Ldeg$.} 
\end{figure}
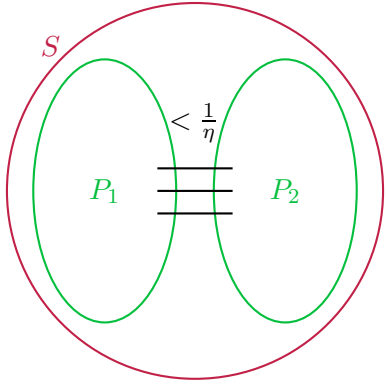

The following definition is key to our reduction to $\Ldeg$-aligned points.

\hypertarget{wellconnected}{\begin{definition}[Well-connected]\label{def:wellconnected}
	Call a set $S \subseteq V$ $\eta$-\emph{well-connected} if $\eta \cdot x_{\mid E(S)} \in \Pstdom(G[S])$, i.e., 
        if $\eta x(\delta_{G[S]}(\cP)) - (|\cP|-1) \geq 0$ for all partitions $\cP$ of $S$.
\end{definition}}

We will make us of the following simple fact, that allows us to contract $\eta$-well-connected subsets of a given set when evaluating the well-connectedness of a given set $S$.

\begin{lemma}\label{lem:contractwellconnected}
Consider a set $S \subseteq V$, and suppose that $S_1, \ldots, S_r$ are disjoint subsets of $S$ that are all $\eta$-well-connected.
Let $G_S=(V_S, E_S)$ be the graph obtained from $G[S]$ after contracting each of $S_1, \ldots, S_r$.
Then $S$ is $\eta$-well-connected, i.e., $\eta \cdot x_{\mid E(S)} \in \Pstdom(G[S])$, if and only if $\eta \cdot x_{\mid E_S} \in \Pstdom(G_S)$.
\end{lemma}
\begin{proof}
    Let $y=\eta x_{\mid E(S)}$.

    First, if $y \in \Pstdom(G[S])$, then certainly $y_{\mid E_S} \in \Pstdom(G_S)$, since given any convex combination of spanning trees of $G[S]$ that dominates $y$, the same convex combination of the images of these spanning trees upon contracting $S_1, \ldots, S_r$ is a convex combination of connected spanning subgraphs of $G_S$ with marginals $y_{\mid E_S}$.

    Conversely, suppose that $y_{\mid E_S} \in \Pstdom(G_S)$.
    If $y_{\mid E_S}=\chi(T)$ for some spanning tree $T$ of $G_S$, and each $y_{\mid E(S_i)} = \chi(T_i)$ for some spanning tree $T_i$ of $S_i$, then the claim is clear; $T \cup T_1 \cup \cdots \cup T_r$ is a spanning tree of $G[S]$.
But the claim clearly remains true upon taking convex combinations, and moreover taking the dominant of any convex combination. 
\end{proof}

\hypertarget{tar:alg}{\begin{algorithm}[h]
\begin{algorithmic}[1]
    \State $\Ldeg' \gets \emptyset$.
    \While{$\Ldeg$ is nonempty}
        \State Choose a minimal set $S \in \Ldeg$.
        \State Let $G_S$ be obtained from $G[S]$ by contracting all the maximal sets in $\Ldeg'$ contained in $S$.
        \State Compute a partition $\cP'$ of $G_S$ minimizing $\eta x(\deltaP_{G_S}(\cP')) - (|\cP'|-1)$.
            Let $\cP$ be the corresponding partition of $S$ obtained by uncontracting. \label{step:partition} 
        \State Delete $S$ from $\Ldeg$ and add all parts of $\cP$ to $\Ldeg'$.
    \EndWhile
    \State Return $\Ldeg'$. 
\end{algorithmic}
\caption{Reduction to a new laminar family}\label{alg:reduction}
\end{algorithm}}

In this section we prove \cref{thm:reduction}, which heavily relies on \cref{alg:reduction}. 
This algorithm will be used to output the new family $\Ldeg'$ in the theorem statement. As such, we first prove some properties of this algorithm.

\begin{lemma}\label{lem:polytime}
    \Cref{alg:reduction} can be implemented in polynomial time.
\end{lemma}
\begin{proof}
    In each iteration, $|\Ldeg|$ decreases, so there are at most $2|V|-1$ iterations.
    Each iteration can be implemented in polynomial time using \Cref{thm:domsep}.
\end{proof}

\begin{lemma}\label{lem:partitionwellconnected}
    Consider any graph $G=(V,E)$ and $\eta > 0$.
    Let $\cP$ be a partition of $G$ that minimizes $\eta x(\deltaP(\cP)) - (|\cP| - 1)$.
    Then each part of $\cP$ is $\eta$-well-connected.
\end{lemma}
\begin{proof}
Fix any part $P \in \cP$, and
    consider any partition 
    $\cQ = \{ Q_1, Q_2, \ldots, Q_r\}$ of $P$.
    Let $\cP'$ be the partition of $V'$ obtained by replacing $P$ with the parts of $\cQ$.

    Write $\delta_P(\cQ)$ for $\delta_{G[P]}(\cQ)$. 
    Since $|\cP'| = |\cP| + |\cQ| - 1$ and $\eta x(\deltaP(\cP')) = \eta x(\deltaP(\cP)) + \eta x(\delta_P(\cQ))$,
    we have
    \[ \eta x(\delta_P(\cQ)) - (|\cQ|-1) = \eta x(\deltaP(\cP')) - (|\cP'|-1) - \bigl( \eta x(\deltaP(\cP)) - (|\cP|-1)\bigr). \]
    This is nonnegative, by our choice of $\cP$,
    and so $\eta x \in \Pstdom(G[P])$.
\end{proof}
\begin{lemma}\label{lem:lamwellconnected}
    The output $\Ldeg'$ of \Cref{alg:reduction} is a laminar family, and each $S \in \Ldeg'$ is $\eta$-well-connected.
\end{lemma}
\begin{proof}
    We claim that throughout the algorithm, we maintain the invariant that $\Ldeg \cup \Ldeg'$ is a laminar family, and that each $S \in \Ldeg'$ is $\eta$-well-connected.
    Certainly this holds at the start of the algorithm.
   Consider a partition  $\cP'$ of $G_S$ generated in step \ref{step:partition}. 
   By  \Cref{lem:partitionwellconnected}, each part of $\cP'$ is $\eta$-well-connected.
   Then since the sets that were contracted in forming $G_S$ are $\eta$-well-connected, by \Cref{lem:contractwellconnected} all parts of $\cP$ are $\eta$-well-connected in $G$.
   Further, no part of $\cP$ crosses a set in $\Ldeg'$, by construction. 
   So the invariant is maintained. 
\end{proof}

The following is the main relevant quality of our reduction. 
\begin{lemma}\label{lem:smallsum}
    Let $S \in \Ldeg$ and let $\Ldeg'$ be the output of \Cref{alg:reduction}. Let $S_1,\dots,S_{\ell}$ be the unique maximal sets in $\Ldeg'$ whose union is $S$. Then, $\sum_{i=1}^{\ell} x(\delta(S_i)) \le \frac{1}{1-\frac{2}{\eta}}x(\delta(S)) - \frac{2}{\eta-2}$. 
\end{lemma}
\begin{proof}
Consider the iteration of the algorithm where $S$ is deleted from $\Ldeg$, and a partition $\cP$ of $S$ (corresponding to a partition $\cP'$ of $G_S$) is added to $\Ldeg'$.
Then $\cP = \{S_1, \ldots, S_\ell\}$.
Note that by the choice of $\cP'$, $\eta x(\deltaP_{G_S}(\cP')) - (|\cP'|-1) \leq 0$ (either $\cP'$ is a violated constraint for $\Pstdom(G_S)$, or if $G_S$ is $\eta$-well-connected, $\cP'$ can be chosen to be the trivial partition of size $1$, and equality is attained).
Converting this to a statement about $\cP$, we have $\eta x(\deltaP_{G[S]}(\cP)) - (|\cP|-1) \leq 0$.
Thus
\begin{align*}
    x(\delta(S)) &= \sum_{i=1}^\ell x(\delta(S_i)) - 2x(\deltaP_{G[S]}(\cP)) \\
                 &\ge \sum_{i=1}^\ell x(\delta(S_i)) - \frac{2}{\eta}(|\cP|-1) \\ 
                 &\ge \left(1-\frac{2}{\eta}\right)\sum_{i=1}^\ell x(\delta(S_i)) +\frac{2}{\eta} &&\text{(as $x(\delta(S_i)) \geq 1$ for each $S_i$).}
\end{align*}
The claim follows.
\end{proof}

We now prove \cref{thm:reduction}.
\reduction*
\begin{proof}
    First, apply \cref{alg:reduction} to $\cL$ to obtain a new family $\Ldeg'$ (which requires only polynomial time by \Cref{lem:polytime}). By \Cref{lem:lamwellconnected}, $\Ldeg'$ is a laminar family of $\eta$-well-connected sets.

We now show that $\eta x$ dominates a point $x'$ which is $\Ldeg'$-aligned, giving i) and ii). 
Let $G_S = (V_S, E_S)$ denote the graph obtained by restricting to $S \in \Ldeg'$ and contracting all children in $\Ldeg'$. 
By definition of $\eta$-well-connected, for any $S \in \Ldeg'$,
$\eta x_{\mid E_S} \in \hyperlink{pstdom}{\Pstdom(G_S)}$. 
It follows that for every $S \in \Ldeg'$ we can find $y_S \in \Pst(G_S)$ with $y_S \leq \eta x_{\mid E_S}$. 
    Combining $y_S$ for each $S$, we obtain $x' \in \Pst(G)$ with $x' \leq \eta x$, and where $x'$ is $\Ldeg$-aligned.
    
    It remains to show (iii).  Fix some $S \in \Ldeg$. The algorithm replaces $S$ by some partition $S_1, S_2, \ldots, S_\ell$ of $S$ in $\Ldeg'$.
    Then we have
    \begin{align*}
        \delt{S}{T} &\le \sum_{i=1}^\ell  \delt{S_i}{T}  && \text{(since $\bigcup_{i=1}^\ell \delta(S_i) \subseteq \delta(S)$)}\\
                    &\le \sum_{i=1}^\ell (\alpha  x'(\delta(S_i)) +\beta) &&\text{(by assumption)}\\ 
                    &\leq \sum_{i=1}^\ell (\eta \alpha x(\delta(S_i) + \beta) && \text{($x' \leq \eta x$)}\\
                    &\leq (\eta \alpha + \beta)\sum_{i=1}^\ell x(\delta(S_i)) 
                            &&\text{(since $x(\delta(S_i)) \geq 1$ for all $S_i$).}
    \end{align*}
    By \cref{lem:smallsum}, $\sum_{i=1}^\ell x(\delta(S_i)) \leq \frac{1}{1-\frac{2}{\eta}} x(\delta(S))$. 
    The claim follows.
\end{proof}

\section{Laminar thin trees for $\Ldeg$-aligned points via iterative relaxation}\label{sec:rounding}

We will now prove \Cref{thm:rounding}, or rather a generalization of it where the graphic matroid is replaced by an arbitrary matroid. First, we define the obvious generalization of $\Ldeg$-aligned for a point in the base polytope of a matroid $\matroid$.



\begin{definition}
    Given a graph $G=(V,E)$, a matroid $\matroid$ with groundset $E$, and a laminar family $\Ldeg$ of $G$, we say that a point $x \in \baseP{\matroid}$ is \emph{$\Ldeg$-aligned} if $x(E(S)) = \rank_\matroid(S)$ for all $S \in \Ldeg$.
\end{definition}
(In the case where $\matroid$ is a graphic matroid, this is just slightly different from the previous definition, if some sets in $\Ldeg$ are not connected. 
The previous definition did not allow for any $\Ldeg$-aligned points in this case, but here it is possible.
This relaxation of the definition is irrelevant; there is no real reason to consider disconnected sets in $\Ldeg$, since they could simply be split into their connected components.)

The following is the primary reason it is useful for a point $x$ to be $\Ldeg$-aligned in the iterative relaxation process.
\begin{lemma}\label{lem:integralES}
Let $x$ be $\Ldeg$-aligned. Let $S \in \Ldeg$ and let $S_1,\dots,S_k \in \Ldeg$ such that $S_i \cap S_j = \emptyset$ for all $1 \le i,j \le k$, $i \neq j$. Let $G_S=(V_S,E_S)$ be the graph arising from contacting $S_1,\dots,S_k$ in the graph $G[S]$. 

Then, $x(E_S)$ is an integer.
\end{lemma}
\begin{proof}
	Since $x(E(S)) = \rank_\matroid(S)$, it is an integer. Similarly, $x(E(S_i))$ is an integer for all $1 \le i \le k$. However $E_S = E(S) \smallsetminus (\cup_{i=1}^k E(S_i))$, from which the claim follows.
\end{proof}

Next, we define the notion of a \emph{matroid} (rather than a point) being $\Ldeg$-aligned. 
\begin{definition}
    Given a graph $G=(V,E)$, a matroid $\matroid$ with groundset $E$, and a laminar family $\Ldeg$ of $G$, we say that 
    $\matroid$ is \emphdef{$\Ldeg$-aligned} if for any basis $B$ of $\matroid$, and every $S \in \Ldeg$, $B \cap E(S)$ is a basis of $\matroid_{\mid E(S)}$.
\end{definition}

The relationship between the notion of a matroid being $\Ldeg$-aligned, and a point $x \in \baseP{\matroid}$ being $\Ldeg$-aligned, is captured by the following lemma.
\begin{lemma}
    A matroid $\matroid$ is $\Ldeg$-aligned if and only if for every point $x \in \baseP{\matroid}$, $x$ is $\Ldeg$-aligned.
\end{lemma}
\begin{proof}
    First suppose $\matroid$ is $\Ldeg$-aligned and let $x \in \baseP{\matroid}$. Then, we can write $x$ as a convex combination of some bases $B_1,\dots,B_k$ of $\matroid$. Since $\matroid$ is $\Ldeg$-aligned, $B_i \cap E(S)$ is a basis of $\matroid_{\mid E(S)}$ for all $S \in \Ldeg$. Thus $|B_i \cap E(S)| = \rnk(\matroid_{\mid E(S)})$ for all $i$. It follows that $x(E(S)) = \rnk(\matroid_{\mid E(S)}) = \rank_\matroid(S)$ for all $S \in \Ldeg$ as desired, demonstrating that $x$ is $\Ldeg$-aligned.
    
    For the other direction, suppose every point $x \in \baseP{\matroid}$ is $\Ldeg$-aligned. 
    Then for any basis $B$ of $\matroid$, by taking $x$ to be the characteristic vector of $B$, we have
    $|B \cap E(S)| = x(E(S)) = \rnk(\matroid_{\mid E(S)})$.
    Thus $\matroid$ is $\Ldeg$-aligned.
\end{proof}

In the previous section, we saw how to reduce to the case where $x$ is a point in the base polytope of the graphic matroid that is $\Ldeg$-aligned.
It will be more convenient for our purposes to work with a matroid that is $\Ldeg$-aligned; this is a stronger property that will ensure that all fractional points we consider later in the iterative relaxation algorithm are all $\Ldeg$-aligned as well.
We can ensure this by \emph{refining} the matroid, in the sense defined in \cite{LOSZ20}.

\begin{definition}
    Given a matroid $\matroid$ and a nonempty proper subset $R$ of the groundset, the \emph{refinement} of $\matroid$ with respect to $R$ is the matroid $\matroid'$ obtained as the direct sum of $\matroid_{\mid R}$ and $\matroid / R$.
\end{definition}
Note that if $\matroid'$ is a refinement of $\matroid$, then every base of $\matroid'$ is a base of $\matroid$.
It is easy to show that for $R \subseteq E$ with $x(R) = \rank_{\matroid}(R)$, $x$ remains in the base polytope of the matroid obtained by refining $\matroid$ with respect to $R$ (see \cite{LOSZ20} for details).
As such, given a point $x \in \baseP{\matroid}$ that is $\Ldeg$-aligned, we can repeatedly refine $\matroid$ by each set of $\Ldeg$ in turn, to obtain a new matroid $\matroid'$ such that $x \in \baseP{\matroid'}$ and $\matroid'$ is $\Ldeg$-aligned.
%
%
For $\matroid$ the graphic matroid, this refinement procedure corresponds to taking $\matroid'$ to be the direct sum of graphic matroids on $G_S$ for each $S \in \Ldeg$.

So we consider the generalization of the laminar thin tree problem to matroids, under the restriction that the matroid is aligned with the laminar family.
An instance of the problem is defined by a graph $G=(V,E)$, a matroid $\matroid$ with groundset $E$, and a laminar family $\Ldeg$ with degree bounds $b_S$ for $S \in \Ldeg$, such that $\matroid$ is $\Ldeg$-aligned.
Edge costs $c_e$ may also be given. 
The goal is to find a minimum cost basis of $\matroid$ satisfying the cut constraints, if a solution exists.

The following LP is the natural relaxation that we will use.
Note that since $\matroid$ is $\Ldeg$-aligned, no explicit additional constraints on $x$ are required; \emph{any} feasible solution must satisfy $x_{\mid E(S)} \in P_{\matroid_{\mid E(S)}}$, and thus must be $\Ldeg$-aligned. 
\begin{equation}\label{eq:roundinglpmat}
\begin{aligned}
	\min \quad& \sum_{e \in E} x_e c_e& \\
	\text{s.t.} \quad & x(\delta(S)) \leq b_S &\forall S\in \Ldeg,\\
	& x \in \Pmat.
\end{aligned}	
\end{equation}

\begin{theorem}[Laminar-constrained matroid basis]\label{thm:rounding-matroid}
    Given an instance \hyperlink{problemstatement-matroid}{$(G, \matroid, \Ldeg, b)$} in which $\matroid$ is $\Ldeg$-aligned, and where the LP relaxation \eqref{eq:roundinglpmat} has a feasible solution $x$, we can find a basis $T$ of $\matroid$ in polynomial time for which $c(T) \leq c(x)$ and 
$\delt{S}{T} \le 2b_S+1$ for all $S \in \Ldeg$. 
\end{theorem} 
\Cref{thm:rounding} is an immediate consequence, by first refining the graphic matroid as described above.

The algorithm we will use to prove this theorem is shown in \Cref{alg:rounding-matroid}. Our algorithm follows the usual iterative relaxation recipe: it ignores edges set to 0 and 1 and then drops constraints which are close to being satisfied. We have one non-standard step which drops a set in $\Ldeg$ if it is \textit{approximately implied} by its immediate parent or child in the family of tight constraints. This non-standard step is what leads to a multiplicative violation instead of an additive one. 

\hypertarget{tar:rounding-alg-matroid}{
    \begin{algorithm}[h]
\begin{algorithmic}[1]
    \Require{Instance $(G=(V,E, c), \matroid, \Ldeg, b)$ where $\matroid$ is tight for $\Ldeg$ and \eqref{eq:roundinglpmat} is feasible.}
\Ensure{Basis $B$ of $\matroid$.}
    \State If $E = \emptyset$, \Return{$\emptyset$}.
    \State Let $x$ be a basic optimal solution to \eqref{eq:roundinglpmat}. 
    \State If there is an edge $e$ with $x_e=0$, \Return{$\text{\sc LamConstrainedBasis}(G - e,\matroid - e, \Ldeg, b)$}. \label{step:delete}
    \State If there is an edge $e$ with $x_e=1$, \Return{$\{e\} \cup \text{\sc LamConstrainedBasis}(G - e,\matroid/e, \Ldeg, b')$}, where $b'_S = b_S$ if $e \notin \delta(S)$, and $b'_S = b_S -1$ if $e \in \delta(S)$. \label{step:contract}
        \State Let $\Ltight$ be the set of cuts $S \in \Ldeg$ with $x(\delta(S)) = b_S$. 
        \State If there is a set $S \in \Ltight$ for which either 
        $\sum_{e \in \delta(S)} (1-x_e) < 3$, or there is an $S' \neq S \in \Ltight$ with $\delta(S') \subseteq \delta(S)$ and $\sum_{e \in \delta(S) \setminus \delta(S')} (1-x_e) <  2$, then \Return{$\text{\sc LamConstrainedBasis}(G,\matroid, \Ldeg \setminus \{S\}, b)$}. \label{step:drop}
        \State \Return{``Fail''}. \Comment{Should not reach this line} \label{step:fail}
\end{algorithmic}
\caption{Procedure \textsc{LamConstrainedBasis}, used to demonstrate \Cref{thm:rounding-matroid}.}
\label{alg:rounding-matroid}
\end{algorithm}
}

%

If a recursive call to \lamconstrainedbasis{} returns ``Fail'', then we consider that the result of the procedure as a whole is also ``Fail''.
We also note that if \lamconstrainedbasis{} is recursively called in any of steps \ref{step:delete}, \ref{step:contract} or \ref{step:drop}, the required properties of the input to the recursive call are satisfied.
In particular, \eqref{eq:roundinglpmat} is feasible. For steps \ref{step:delete} and \ref{step:contract}, $x_{\mid E -e}$ is feasible for the smaller instance; for step \ref{step:drop}, simply $x$ is.
With this in mind, \lamconstrainedbasis{} is well-defined.

We first show that as long as the algorithm does succeed, the returned basis obeys the theorem statement.
\begin{lemma}
If \Cref{alg:rounding-matroid} does not return ``Fail'', the returned set $B$ is a basis and obeys $c(B) \le c(x)$ and $\delt{S}{B} \le 2b_S+1$ for all $S \in \Ldeg$. 
\end{lemma}
\begin{proof}
We prove the claim by induction on $|E| + |\Ldeg|$.
The claim is trivially true if $E = \emptyset$.

So suppose the claim holds for all smaller values of $|E|+|\Ldeg|$.
If $x_e = 0$ for some $e$ in step \ref{step:delete}, then the claim is immediate; as long as the recursive call succeeds, returning a basis $B'$ of $\matroid-e$ approximately satisfying the constraints, then $B=B'$ is of course a basis of $\matroid$ still approximately satisfying the constraints. Furthermore, since $c(B') \le c(x')$ where $x'$ is a basic optimal solution to the problem on $\matroid-e$, and $x_{\mid E \setminus \{e\}}$ is feasible for the problem on $\matroid-e$, $c(B) = c(B') \le c(x') \le c(x)$. 
If $x_e = 1$ for some $e$ in step \ref{step:contract}, and the recursive call succeeds and returns a basis $B'$ of $\matroid/e$, then $B := B' \cup \{e\}$ is a basis of $\matroid$.
Further, for any set $S \in \Ldeg$ with $e \notin \delta(S)$, we have 
\[ |B \cap \delta(S)| = |B' \cap \delta(S)| \leq 2b'_S + 1 = 2b_S + 1. \]
On the other hand if $e \in \delta(S)$, we have
\[ |B \cap \delta(S)| = |B' \cap \delta(S)| + 1 \leq 2b'_S + 2 < 2b_S + 1. \]
Finally, since $c(B') \le c(x')$ where $x'$ was a basic optimal solution to the problem on $\matroid/e$, and $x_{\mid E\setminus \{e\}}$ is feasible for the problem on $\matroid/e$, $c(B) = c(B')+c(e) \le c(x')+c(e)\le c(x)$.

It remains to consider the situation where we drop a constraint in step \ref{step:drop}.
Suppose a set $S \in \Ltight$ is dropped because $|\delta(S)| - x(\delta(S)) = \sum_{e \in \delta(S)}(1-x_e) < 3$.
Since the constraint is tight, we deduce that $|\delta(S)| - b_S < 3$, and so $|\delta(S)| \leq b_S + 2 \leq 2b_S+1$ as desired.

Now suppose $S \in \Ltight$ is dropped because there is an $S' \neq S \in \Ltight$ with $\delta(S') \subseteq \delta(S)$ and $\sum_{e \in \delta(S) \setminus \delta(S')}(1-x_e) < 2$. 
By tightness, $x(\delta(S)\setminus \delta(S')) = x(\delta(S)) - x(\delta(S')) = b_S - b_{S'}$ is an integer.
Note that either $\delta(S') = \delta(S)$, in which case clearly we can drop the duplicate constraint, or $b_S > b_{S'}$; assume the latter.
We have $|\delta(S) \setminus \delta(S')| \leq 1 + b_S - b_{S'}$.
Suppose $B$ is any basis satisfying $|B \cap \delta(S')| \leq 2b_{S'}+1$.
Then 
\begin{align*} |B \cap \delta(S)| &\leq |B \cap \delta(S')| + |\delta(S) \setminus \delta(S')|\\
&\leq (2b_{S'} + 1) + 1 + b_S - b_{S'}\\
&= b_{S'} + 2 + b_S\\
&\leq (b_S - 1) + 2 + b_S \leq 2b_S + 1.
\end{align*}
Of course, dropping a constraint can only decrease the cost of a basic optimal solution to \eqref{eq:roundinglpmat}, so $c(B) \leq c(x)$ is immediate by induction in this case.
\end{proof}

%
%

Now we are ready to prove the theorem. 
\begin{proof}[Proof of \Cref{thm:rounding-matroid}]
By the above lemma, it is enough to prove that the algorithm succeeds.
For this, it suffices to show that whenever the preconditions of \lamconstrainedbasis{} are satisfied, the procedure never reaches step \ref{step:fail}.

Suppose for a contradiction that we do reach step \ref{step:fail}.
By assumption, none of the constraints defining the extreme point $x$ are of the form $x_e = 0$ or $x_e = 1$, so they all come from tight cut constraints and tight matroid constraints.
Let $\Cbasis = \{C_1, C_2, \ldots, C_r\}$, with $C_1 \subsetneq C_2 \cdots \subsetneq C_r \subseteq E$ and $\LdegB \subseteq \Ltight$ 
be such that the constraints 
$x(\delta(S)) = b_S$ for $S \in \LdegB$ and $x(C) = \rnk_\matroid(C)$ for $C \in \Cbasis$ are a collection of linearly independent tight constraints defining $x$. 
Moreover, choose this basis of tight constraints in such a way that $|\Cbasis|$ is as large as possible.
The fact that the tight matroid constraints form a chain follows from standard uncrossing arguments (see \cite{Sch03} Chapter 41 or \cite{KLS08}). 
Since there are precisely $|E|$ defining constraints, we have
\[ |E| = |\LdegB| + |\Cbasis|. \] 
We note that since $\matroid$ is $\Ldeg$-aligned, the maximality of $\Cbasis$ ensures that $E(S) \in \Span(\Cbasis)$ for each $S \in \Ldeg$.

Assign 1 splittable token to each $e \in E$; our goal will be to assign these tokens to the constraints of $\LdegB$ and $\Cbasis$ so that each tight constraint gets 1 token, and there is something left over. This will be our desired contradiction.

We will assign $x_e$ tokens to $C_i$ for each $e \in C_i \setminus C_{i-1}$.
Since $0 < x_e < 1$ for each $e$, and $x(C_i)$ and $x(C_{i-1})$ are both integers with $x(C_{i-1}) < x(C_i)$, we can deduce that $x(C_i \setminus C_{i-1}) \geq 1$. 

Now each edge has $1-x_e$ tokens remaining. Our token assignment scheme will be as follows.
We start with an assignment that is very reminiscent of the scheme for degree bounded spanning trees~\cite{SL15}.
For each $e=\{u,v\}$, we assign $(1-x_e)/2$ tokens to the smallest set in $\LdegB$ containing $u$, and $(1-x_e)/2$ tokens to the smallest set in $\LdegB$ containing $v$.
After this, we work bottom up on $\LdegB$, and if $S \in \LdegB$ has strictly more than the 1 token needed, we assign the excess to its parent in $\LdegB$.

First, any minimal set $S \in \LdegB$ satisfies $\sum_{e \in \delta(S)} (1-x_e) \geq 3$, meaning that at least $\frac{3}{2}$ tokens are initially assigned to $S$.
So $S$ receives enough tokens to give a half token as excess to its parent.
Inductively, we claim that every set gets 1 token, and moreover, has an excess of at least $\frac{1}{2}$ that it can give to its parent.
For any non-minimal $S \in \LdegB$, we have three cases depending on the number of disjoint maximal children of $S$ in $\LdegB$. In each case we will consider the graph $G_S=(V_S,E_S)$ resulting from contracting the maximal children of $S$ in $\LdegB$ in the graph $G[S]$. In Cases 2 and 3 we crucially use that $x(E_S)$ is an integer by \cref{lem:integralES}.

\begin{itemize}[-]
    \item \textbf{Case 1: $S$ has at least three maximal children in $\LdegB$.} 

        Then inductively, each of these children has an excess of at least $\frac{1}{2}$.
        This gives us at least $\frac{3}{2}$ tokens for $S$, as desired.
        
        \medskip

	\item \textbf{Case 2: $S$ has exactly two maximal children $A,B \in \LdegB$.}

	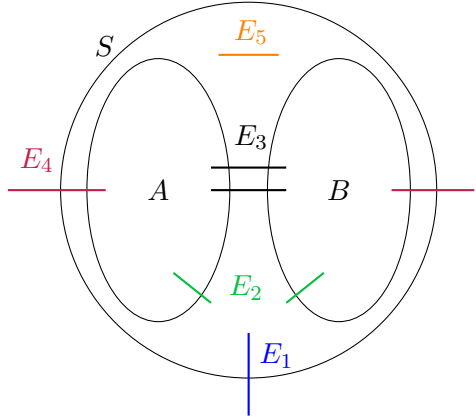
\begin{figure}[htb!]
\centering
\begin{tikzpicture}
\draw (0,0) node[above left=1.65cm and 1.65cm] {$S$} circle (2.5cm);
\draw (-1.2,0) node[ellipse, draw=black, minimum width = 1.9cm, 
    minimum height = 3.5cm] (C) {$A$};
\draw (1.2,0) node[ellipse, draw=black, minimum width = 1.9cm, 
    minimum height = 3.5cm] (C) {$B$};
\draw[-, thick] (-0.5,0.3) -- (0.5,0.3) node[above left=0.1cm and 0.1cm] {$E_3$};
\draw[-, thick] (-0.5,0) -- (0.5,0) node {};

\draw[-, thick,blue] (0,-3) -- (0,-1.9) node[below right,text=blue] {$E_1$};

\draw[-, thick, darkgreen] (-1,-1.1) -- (-0.5,-1.5) node[above right=-0.1cm and 0.1cm,text=darkgreen] {$E_2$};
\draw[-, thick, darkgreen] (1,-1.1) -- (0.5,-1.5) node {};

\draw[-, thick, orange] (-0.4,1.8) -- (0.4,1.8) node[above left,text=orange] {$E_5$};

\draw[-, thick, darkred] (-3.2,0) -- (-1.9,0) node[text=red] {} node[above left=0.1cm and 0.55cm] {$E_4$};;
\draw[-, thick, darkred] (3,0) -- (1.9,0) node {};
\end{tikzpicture}
\caption{Setting for Case 2. Note some edge sets may be empty.}\label{fig:case2}
\end{figure}

	Inductively, each child has an excess of at least $\frac{1}{2}$, giving us at least one token. Thus we need to collect at least $\frac{1}{2}$ additional tokens. 
	
	Consider the edge sets as defined in \cref{fig:case2}. In particular,
	\begin{align*}
		E_1 &= \delta(S) \smallsetminus (\delta(A) \cup \delta(B)) \\
		E_2 &= (\delta(A) \symdiff \delta(B)) \smallsetminus \delta(S) \\
		E_3 &= \delta(A) \cap \delta(B) \\
		E_4 &= (\delta(A) \cup \delta(B)) \cap \delta(S) \\
		E_5 &= E_S \smallsetminus (E_2 \cup E_3)
	\end{align*}
	First we observe that $E_1 \cup E_2 \cup E_5$ is nonempty. 
        For suppose not; then, with $\chi$ denoting the incidence vector of a set,
        we can write 
        \[ \indset{\delta(S)} + 2\indset{E(S)} = \indset{\delta(A)} + \indset{\delta(B)} + 2\indset{E(A)} + 2\indset{E(B)}. \]
    However, by the maximality of our choice of $\Cbasis$, $E(A)$, $E(B)$ and $E(S)$ are all in the span of $\Cbasis$, whereas $\delta(S)$, $\delta(A)$ and $\delta(B)$ are all in $\LdegB$. 
    Thus we have a linear dependence among the constraints defined by $\Cbasis$ and $\LdegB$, a contradiction.

    So $x(E_1) + x(E_2) + x(E_5) > 0$. Therefore, we get
	\begin{align*}
		\frac{|E_1|+|E_2|-x(E_1)-x(E_2)}{2}+|E_5|-x(E_5) &= z-\left(\frac{x(E_1)+x(E_2)}{2}+x(E_5)\right) > 0
	\end{align*}
        fractional tokens for some $z \in \Z_{\geq 0}$. We will prove that $\frac{x(E_1)+x(E_2)}{2}+x(E_5)$ is half integral, 
        from which the claim follows. By the integrality of $x(E_S)$ (using $\Ldeg$-alignment) and the tightness of the constraints on $A,B$ and $S$, we have that 
	\begin{align*}
            a := x(E_2)+x(E_3)+x(E_5), \hspace{2mm} b := x(E_2)+2x(E_3)+x(E_4) \hspace{2mm} \text{and} \hspace{2mm} c := x(E_1)+x(E_4)
	\end{align*}
        are all integers.
        Since $a-b/2 + c/2 = 
        \frac{x(E_1)+x(E_2)}{2}+x(E_5)$, the claim follows.
	
		\medskip 
    \item \textbf{Case 3: $S$ has precisely one maximal child $S'$ in $\LdegB$.}

        	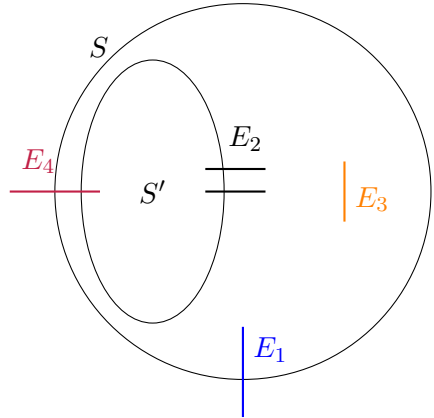
\begin{figure}[htb!]
\centering
\begin{tikzpicture}
\draw (0,0) node[above left=1.65cm and 1.65cm] {$S$} circle (2.5cm);
\draw (-1.2,0) node[ellipse, draw=black, minimum width = 1.9cm, 
    minimum height = 3.5cm] (C) {$S'$};
\draw[-, thick] (-0.5,0.3) -- (0.3,0.3) node[above left=0.1cm and -0.1cm] {$E_2$};
\draw[-, thick] (-0.5,0) -- (0.3,0) node {};

\draw[-, thick,blue] (0,-3) -- (0,-1.8) node[below right,text=blue] {$E_1$};

\draw[-, thick, orange] (1.35,0.4) -- (1.35,-0.4) node[above right,text=orange] {$E_3$};

\draw[-, thick, darkred] (-3.1,0) -- (-1.9,0) node[text=red] {} node[above left=0.1cm and 0.45cm] {$E_4$};;
\end{tikzpicture}
\caption{Setting for Case 3.}\label{fig:case3}
\end{figure}

        We need to find 1 token that has been given by edges directly to $S$, so that the $\frac{1}{2}$ excess token from $S'$ can be carried over as the excess of $S$.

        If $\delta(S') \subseteq \delta(S)$, then because no relaxation step was possible in line \ref{step:drop}, $\sum_{e \in \delta(S) \setminus \delta(S')}(1-x_e) \geq 2$. Since each edge in $\delta(S) \setminus \delta(S')$ contributes $(1-x_e)/2$ tokens, this gives us our token as necessary.
        Similarly, if $\delta(S') \supseteq \delta(S)$ we get the desired one token.

        So assume that $ \delta(S) \setminus \delta(S')$ and $\delta(S') \setminus \delta(S)$ are both nonempty.
        Let 
        \begin{align*}
            E_1 &:= \delta(S) \setminus \delta(S'),\\
            E_2 &:= \delta(S') \setminus \delta(S),\\
            E_3 &:= E_S \setminus \delta(S'), \text{ and}\\
            E_4 &:= \delta(S) \cap \delta(S').
        \end{align*}
        (See \Cref{fig:case3}.)

        Let $\delta \in [0,1)$ be the fractional part of $x(E_4)$.
        Note that the number of tokens assigned to $S$ is
        \begin{equation}\label{eq:tokencount} 
            |E_3| - x(E_3) + \tfrac12(|E_1| + |E_2| - x(E_1) - x(E_2)).
        \end{equation}
        Also observe that 
    \begin{equation}\label{eq:integral}
        x(E_2) + x(E_3), \quad x(E_1) + x(E_4), \quad \text{and} \quad x(E_2) + x(E_4)
    \end{equation}
        are all integer-valued, by tightness of the cut constraints and $\Ldeg$-alignment.
        We distinguish two subcases.
        \begin{itemize}
            \item $\delta = 0$. Then $x(E_1)$ and $x(E_2)$ are both integers, and moreover since $E_1$ and $E_2$ are nonempty, $|E_1| - x(E_1)$ and $|E_2| - x(E_2)$ are both positive integers. This already gives us the desired 1 token by \eqref{eq:tokencount}. 

            \item $\delta > 0$. Then by \eqref{eq:integral} the fractional parts of $x(E_1)$ and $x(E_2)$ are both $1-\delta$, and the fractional part of $x(E_3)$ is then $\delta$.
                Thus $|E_1| - x(E_1) \geq \delta$ (being positive, with fractional part $\delta$); similarly, $|E_2| - x(E_2) \geq \delta$ and $|E_3| - x(E_3) \geq 1-\delta$.
                Substituting into \eqref{eq:tokencount}, we have at least $1-\delta + (2\delta)/2 = 1$ tokens assigned to $S$, as required.
        \end{itemize}
    \end{itemize}

    We have demonstrated that all sets in $\LdegB$ receive a full tokens; moreover, any maximal set in $\LdegB$ will have an extra token that is not needed, since it has no parent to give it to.
    So we have our desired contradiction: $|E| > |\Cbasis| + |\LdegB|$. 
   \end{proof}

\section{Conclusion}\label{sec:conclusion}

Besides the (strong) thin tree conjecture, our work leaves open several directions. One fascinating question is whether it is possible to leverage or strengthen our results to give a novel constant factor approximation algorithm for ATSP. While an algorithmic version of the strong thin tree conjecture is sufficient to give a constant factor approximation algorithm for ATSP, it is unclear if it is necessary: indeed, the current constant factor approximation algorithm for ATSP is not known to imply anything about thin trees. We ask if perhaps it is sufficient to focus on thinness for a laminar (or near laminar) family of cuts. 

 A second open question is whether it is possible to achieve a \textit{minimum cost} tree which violates the degree bounds in a laminar family by any constant factor. One would need to avoid the scaling currently present in our reduction. A natural relaxation of this question is to ask for a $1+\epsilon$ approximation for arbitrarily small $\epsilon$ as has been done for the chain case \cite{LS16}.
 
 Finally, we note that our results immediately give a thin tree with respect to the set of minimum cuts of any graph, and we believe it may be possible to extend it to the set of all $(1+\epsilon)$ \textit{near} minimum cuts for some small $\epsilon > 0$ using results from \cite{KKO21b}. We ask whether it is possible to extend our result to more general families of cuts such as the union of a constant number of laminar families or the set of cuts with at most $\alpha k$ edges in the graph for some constant $\alpha$ significantly larger than 1. 

\paragraph{Acknowledgments.} 

We thank Shayan Oveis Gharan for asking about thin trees for minimum cuts, discussions of which were crucial in the early stages of this project. 
We thank David Shmoys and Mohit Singh for asking the question about ATSP now stated in the conclusion.
N.O.\ thanks Michel Goemans and Rico Zenklusen for varied discussions on the topic of thin trees.
This work began during the Trimester Program on Combinatorial Optimization at 
the Hausdorff Research Institute for Mathematics;
we gratefully acknowledge the institute for its generous hosting and support.

\printbibliography

\end{document}